\documentclass{llncs}
\usepackage{hyperref}
\usepackage{subfig}
\usepackage{graphicx} %
\usepackage{booktabs}
\usepackage{times}
\usepackage{etex}
\usepackage{mathtools}
\PassOptionsToPackage{table}{xcolor}
\usepackage{amsmath}
\usepackage{amssymb}
\usepackage{xspace}
\usepackage{graphicx}
\usepackage{colortbl}
\usepackage{etoolbox}%
\AtEndEnvironment{proof}{\qed}%
\newcommand{\citep}[1]{\cite{#1}\xspace}

\newcommand{\Ptime}{{\sc P\xspace}}
\newcommand{\NP}{{\sc NP}\xspace}

\newcommand{\subss}{{\sc Subset Sum}\xspace}
\newcommand{\usubss}{{\sc Unbounded Subset Sum}\xspace}
\newcommand{\knapsack}{{\sc Knapsack}\xspace}
\newcommand{\uknapsack}{{\sc Unbounded Knapsack}\xspace}
\newcommand{\partition}{{\sc Partition}\xspace}

\newcommand{\Rat}{\ensuremath{\mathbb{Q}}\xspace}
\newcommand{\calO}{\mathcal{O}}
\newcommand{\nfrac}[2]{\frac{#1}{#2}}
\newcommand{\plone}[1]{\oplus_{#1} 1}
\newcommand{\minone}[1]{\ominus_{#1} 1}
\newcommand{\ocnum}{4\xspace}
\newcommand{\sameSAT}{{\sc All-the-Same-SAT}\xspace}
\newcommand{\XSAT}{{\sc One-in-Three-SAT}\xspace}
\newcommand{\SAT}{{\sc Satisfiability}\xspace}
\newcommand{\cnffour}{3-CNF$_{\leq \ocnum}$\xspace}

\newcommand{\val}{\mathsf{profit}}
\newcommand{\OPT}{\mathrm{OPT}}

\begin{document}
\title{On Strong NP-Completeness of Rational Problems}
\author{Dominik Wojtczak}
\institute{University of Liverpool, UK}
\maketitle
\begin{abstract}
The computational complexity of 
the partition, 0-1 subset sum, unbounded subset sum, 0-1 knapsack and unbounded knapsack problems 
and their multiple variants
were studied in numerous papers in the past 
where all the weights and profits were assumed to be integers.
We re-examine here the computational complexity of all these problems
in the setting where the weights and profits are allowed 
to be any rational numbers. 
We show that all of these problems in this setting become 
strongly NP-complete and, as a result,
no pseudo-polynomial algorithm can exist for solving them
unless \Ptime=\NP. Despite this result we show that 
they all still admit a fully polynomial-time approximation scheme.

\end{abstract}

\section{Introduction}

The problem of partitioning a given set of items into two parts with equal total weights (that we will refer to as \partition) 
goes back at least to 1897 \cite{mathews1897partition}.
A well-known generalisation is the problem of finding a subset with a given total weight (0-1 \subss) and
the same problem where each item can be picked more than once (\usubss).
Finally, these are commonly generalised to the setting where each item also has a profit and
the aim is to pick a subset of items with the total profit higher than a given threshold, 
but at the same time their total weight smaller than a given capacity (0-1 \knapsack).
A variant of the last problem where each item can be picked more than once is also studied (\uknapsack).

The \subss problem has numerous applications: its solutions can be used for designing better lower bounds for scheduling problems (see, e.g., \cite{hoogeveen1994new} and \cite{gueret1999new}) and it appears as a subproblem in numerous combinatorial problems (see, e.g., \cite{pisinger1999exact}). At the same time, many industrial problems can be formulated as knapsack problems: cargo loading, cutting stock, 
capital budgeting, portfolio selection, interbank clearing systems, knapsack cryptosystems, and combinatorial auctions 
to name a couple of examples (see Chapter 15 in \cite{books/daglib/0010031} for more details regarding these problems and their solutions).

The decision problems studied in this paper were among the first ones to be shown to be NP-complete \cite{karp1972reducibility}. 
At the same they are considered to be the easiest problems in this class, 
because they are polynomial time solvable if items' weights and profits are represented using the unary notation
(in other words, they are only {\em weakly} NP-complete). 
In particular, they can be solved in polynomial time when these numbers are bounded by a fixed constant
(and the number of items is unbounded).
Furthermore, the optimisation version of all these decision problems admit {\em fully polynomial-time approximation schemes (FPTAS)}, i.e.,
we can find a solution with a value at least equal to $(1 - \epsilon)$ times the optimal in time polynomial in the size of the input and ${1}/{\epsilon}$ for any $\epsilon > 0$.

To the best of our knowledge, the computational complexity analysis of all these problems was only studied so far 
under the simplifying assumption that all the input values are integers.
However, in most settings where these problems are used, these numbers are likely to be given as rational numbers instead.
We were surprised to discover that the computational complexity in such a rational setting was not properly studied before.
Indeed, as pointed out in \citep{books/daglib/0010031}:
\begin{quote}
	A rather subtle point is the question of rational coefficients. 
	Indeed, most textbooks get rid of this case, where some or all input values are non-integer, 
	by the trivial statement that multiplying with a suitable factor, 
	e.g. with the smallest common multiple of the denominators, 
	if the values are given as fractions or by a suitable power of 10, 
	transforms the data into integers.
	Clearly, this may transform even a problem of moderate size into a rather unpleasant problem with huge coefficients.
\end{quote}
This clearly looks like a fundamental gap in the understanding of the complexity of these computational problems.
Allowing the input values to be rational makes a lot of sense in many settings.
For example, we encountered this problem when studying an optimal control in multi-mode systems with discrete costs \cite{DBLP:conf/time/MousaSW16,FORMATS17/MSW17}
and looking at the weighted voting games (see, e.g., \cite{elkind2009computational} where the weights are defined to be rational).
An interesting real-life problem is checking whether we have the exact amount when paying, which is important in a situation when no change can be given.
While we take decimal monetary systems for granted these days, there were plenty of non-decimal monetary systems in use not so long ago.
For example, in the UK between 1717 and 1816 one pound sterling was worth twenty shillings, one shilling was worth twelve pence, and one guinea was worth twenty one shillings.

We show here that allowing the input numbers to be rational makes a significant difference and in fact all these decision problems in such a setting become strongly \NP-complete \cite{garey1978strong}, i.e., they are \NP-complete even when all their numerical values are at most polynomial
in the size of the rest of the input or, equivalently, if all these numerical values are represented in unary.
To prove this we will show an \NP-completeness of a new variant of the {\sc satisfiability} problem and
use results regarding distribution of prime numbers.
As a direct consequence of our result, there does not exist any pseudo-polynomial algorithms for solving these decision problems unless \Ptime=\NP.
At the same time, we will show that they still all admit a fully polynomial-time approximation scheme (see, e.g., \cite{books/daglib/0010031}).
This may seem wrong, because the paper that introduced strong \NP-completeness 
\cite{garey1978strong} also showed that
no strongly \NP-hard problem can admit an FPTAS unless \Ptime=\NP. 
However, the crucial assumption made there is that the objective function is integer valued,
which does not hold in our case.

\medskip\noindent{\bf Related work. }
The decision problems studied in this paper are so commonly used that they have already been thousands of papers published
about them and many of their variants, including multiple algorithms and heuristic for solving them precisely and approximately.
There are also two full-length books, \cite{silvano1990knapsack} and \cite{books/daglib/0010031}, solely dedicated to these problems. 

Several extensions of the classic knapsack problem were shown to be strongly NP-complete.
These include partially ordered knapsack \cite{johnson1983knapsacks} (where we need to pick a set of items closed under predecessor), graph partitioning \cite{johnson1983knapsacks} (where we need to partition a graph into $m$ disjoint subgraphs under cost constraints),
multiple knapsack problem \cite{chekuri2005polynomial}, knapsack problem with conflict graphs \cite{pferschy2009knapsack} (where we restrict which pairs of items can be picked together), and quadratic knapsack problem \cite{gallo1980quadratic,books/daglib/0010031} 
(where the profit of packing an item depends on how well it fits together with the other selected items).

The first FPTAS for the optimisation version of the \knapsack problem was established in 1975 by Ibarra and Kim \cite{ibarra1975fast} 
and independently by Babat \cite{babat1975linear}. Multiple other, more efficient, FPTAS for these problems followed (see, e.g., \cite{books/daglib/0010031}).

\medskip\noindent{\bf Plan of the paper.}
In the next section, we introduce all the used notation as well as formally define all the decision problems that we study in this paper.
In Section \ref{sec:prime}, we analyse the amount of space one needs to write down the first $n$ primes in the unary notation
as well as a unique representation theorem concerning sums of rational numbers.
In Section \ref{sec:sat}, we define a couple of new variants of the well-known satisfiability problem for Boolean formulae in 3-CNF form
and show them to be \NP-complete.
Our main result, concerning the strong NP-hardness of all the decision problems studied in this paper 
with rational inputs, can be found in Section \ref{sec:strong} 
and it builds on the results from Sections \ref{sec:prime} and \ref{sec:sat}.
We briefly discuss the existence of FPTAS for the optimisation version of our decision problems in Section \ref{sec:fptas}.
Finally, we conclude in Section~\ref{sec:conclusions}.

\section{Background}
\label{sec:background}

Let $\Rat_{\geq 0}$ be the set of non-negative rational numbers.
We assume that a non-negative rational number is represented as usual as
a pair of its numerator and denominator, 
both of which are natural numbers that do not have a common divisor greater than $1$. 
A unary representation of a rational number is simply a pair of its numerator and denominator represented in unary.
For two natural numbers $a$ and $n$, let $a \!\!\mod n \in \{0,\ldots,n-1\}$ 
denote the remainder of dividing $a$ by $n$.
For any two numbers $a,b \in \{1,\ldots,n\}$, 
we define their addition modulo $n$, denoted by $\oplus_{n}$, as follows:
$a \oplus_{n} b = ((a + b - 1) \!\!\mod n) + 1$.
Note that we subtract and add 1 in this expression so that
the result of this operation belongs to $\{1,\ldots,n\}$.
Similarly we define the subtraction modulo $n$, denoted by $\ominus_{n}$, as follows:
$a \ominus_{n} b = ((n + a - b - 1) \!\!\mod n) + 1$.
We assume that $\oplus_{n}$ and $\ominus_{n}$ operators have higher precedence
than the usual $+$ and $-$ operators.

We now formally define all the decision problems that we study in this paper.

\begin{definition}[\subss problems]
	Assume we are given a list of $n$ items with rational non-negative weights $A = \{w_1, \ldots, w_n\}$ 
	and a target total weight $W \in \Rat_{\geq 0}$. 
	
	0-1 \subss: 
	Does there exists a subset $B$ of $A$ such that 
	the total weight of $B$ is equal to $W$?
	
	\usubss:
	Does there exist a list of non-negative integer quantities 
	$(q_1, \ldots, q_n)$ such that 
	$$\sum_{i=1}^n q_i \cdot w_i = W ? $$
	(Intuitively, $q_i$ denotes the number of times the $i$-th item in $A$ is chosen.)
\end{definition}

A natural generalisation of this problem where each item gives us a profit when picked
is the well-known knapsack problem.

\begin{definition}[\knapsack problems]
	Assume there are $n$ items whose non-negative rational weights and profits are given
	as a list $L = \{(w_1,v_1), \ldots, (w_n,v_n)\}$.
	Let the capacity be $W \in \Rat_{\geq 0}$ and the profit threshold be $V \in \Rat_{\geq 0}$.
	
	0-1 \knapsack: 
	Is there a subset of $L$ whose
	total weight does not exceed $W$
	and total profit is at least $V$?	
	
	\uknapsack:
	Is there a list of non-negative integers
	$(q_1, \ldots, q_n)$ such that 
	$$\sum_{i=1}^n q_i \cdot w_i \leq W \text{\qquad and \qquad} \sum_{i=1}^n q_i \cdot v_i \geq V?$$ 
	(Intuitively, $q_i$ denotes the number of times the $i$-th item in $A$ is chosen.)
	
\end{definition}

Finally, a special case of the \subss problem is the \partition problem.

\begin{definition}[\partition problem]
	Assume we are given a list of $n$ items with non-negative rational weights $A = \{w_1, \ldots, w_n\}$.
	
	Can the set $A$ be partitioned into two sets with equal total weights?
	
\end{definition}

Now let us compare the size of a \partition problem instance when represented in binary and unary notation.
Let $A = \{w_1, \ldots, w_n\}$ be an instance such that $w_i = \frac{a_i}{b_i}$ where $a_i, b_i \in \mathbb{N}$ for all $i=1,\ldots,n$.
Notice that the size of $A$ is $\Theta\big(\sum_{i=1}^n \log(a_i) + \sum_{i=1}^n\log(b_i)\big)$ when written down in binary 
and $\Theta\big(\sum_{i=1}^n a_i + \sum_{i=1}^n b_i\big)$ in unary. 
If we now multiply all weights in $A$ by $\prod_{i=1}^n b_i$ then we would 
get an equivalent instance $A'$ with only integer weights. 
The size of $A'$ would be 
$\Theta\big(\sum_{i=1}^n \log(a_i/b_i\prod_{j=1}^n b_j)\big) = \Theta\big(\sum_{i=1}^n \log a_i + (n-1)\sum_{i=1}^n \log b_i\big)$ 
when written down in binary and in unary: $\Theta\big(\sum_{i=1}^n (a_i/b_i\prod_{j=1}^n b_j)\big) = \Omega(\min a_i \cdot (\min b_i)^{n-1})$. Notice that the first expression is polynomial in the size of the original instance while the second one
may grow exponentially. A similar analysis shows the same behaviour for all the other decision problems studied in this paper.

\section{Prime Suspects}
\label{sec:prime}

In this section we first show that writing down all the first $n$ prime numbers in the unary notation can be done
using space polynomial in $n$. Let $\pi_i$ denote the $i$-th prime number. The following upper bound is known for $\pi_i$.

\begin{theorem}[inequality (3.13) in \cite{rosser1962approximate}]
	\label{thm:prime-dist}
	$$\pi_i < i(\log i + \log \log i) \text{\ \ \ \ for\ \ \ \ } i \geq 6$$
\end{theorem}

This estimate gives us the following corollary that will be used in the main result of this paper.

\begin{corollary}
	\label{cor:primes-complexity}
	The total size of the first $n$ prime numbers, when written down in unary, is $\calO(n^2\log n)$.
	Furthermore, they can be computed in polynomial time.
\end{corollary}
\begin{proof}
	Let $n \geq 6$, because otherwise the problem is trivial.
	Thanks to Theorem \ref{thm:prime-dist}, it suffices to list all natural numbers smaller than $2n\log n$
	(because $n(\log n + \log \log n) \leq 2n\log n$)
	and use the sieve of Eratosthenes to remove all nonprime numbers from this list.
	It follows that writing down the first $n$ prime numbers requires $\calO(n^2\log n)$ space.
	The sieve can easily be implemented in polynomial time and, to be precise,
	in this case $\calO(n^2\log^2 n)$ additions and $\calO(n\log n)$ bits of memory would suffice.
\end{proof}

Now we prove a result regarding a unique representation of rational numbers expressed as sums of fractions with prime denominators, which in a way is quite similar to the Chinese remainder theorem.

\begin{lemma}
	\label{lem:uniqueness}
	Let $(p_1, \ldots, p_n)$ be a list of $n$ different prime numbers.
	Let	$(a_0,a_1, \ldots, a_n)$ and $(a_0,b_1,\ldots,b_n)$ 
	be two lists of integers such that
	$|a_i - b_i| < p_i$ holds for all $i = 1,\ldots,n$. 
	We then have
	$$a_0 + \frac{a_1}{p_1} + \ldots + \frac{a_n}{p_n} = b_0 + \frac{b_1}{p_1} + \ldots + \frac{b_n}{p_n}$$
	$$\text{ if and only if }$$ 
	$$a_i = b_i \text{\ \ for all\ \ } i = 0,\ldots,n.$$
\end{lemma}
\begin{proof}
	$(\Leftarrow)$
	If $a_i = b_i \text{ for all } i = 0,\ldots,n$ holds then obviously
	$$a_0 + \frac{a_1}{p_1} + \ldots + \frac{a_n}{p_n} = b_0 + \frac{b_1}{p_1} + \ldots + \frac{b_n}{p_n}.$$
	
	$(\Rightarrow)$
	We need to consider two cases: $a_0 = b_0$ and $a_0 \neq b_0$. 
	In the first case, suppose that
	$a_j \neq b_j$ for some $j \in \{1,\ldots,n\}$. If we multiply
	$$\frac{a_1-b_1}{p_1} + \ldots + \frac{a_n-b_n}{p_n}
	\text{\ \ \ by\ \ \ } \prod_{i=1}^{n} p_i$$ then we would get an integer, which is not divisible by $p_j$ 
	(because $0 < |a_j-b_j| < p_j$)
	and so this expression cannot be equal to $0$.
	Therefore, in this case, 
	$$a_0 + \frac{a_1}{p_1} + \ldots + \frac{a_n}{p_n} \neq b_0 + \frac{b_1}{p_1} + \ldots + \frac{b_n}{p_n}.$$
	In the second case, if $a_i = b_i \text{ for all } i = 1,\ldots,n$ holds
	then clearly 
	$$a_0 + \frac{a_1}{p_1} + \ldots + \frac{a_n}{p_n} \neq b_0 + \frac{b_1}{p_1} + \ldots + \frac{b_n}{p_n}.$$
	Otherwise, again suppose that
	$a_j \neq b_j$ for some $j \in \{1,\ldots,n\}$.
	If we multiply
	$$a_0 - b_0 + \frac{a_1-b_1}{p_1} + \ldots + \frac{a_n-b_n}{p_n}
	\text{\ \ \ by\ \ \ } \prod_{i=1}^{n} p_i$$ then we would get an integer, which is not divisible by $p_j$ 
	(because $0 < |a_j-b_j| < p_j$)
	and so this expression cannot be equal to $0$.
	Therefore, again, in this case, 
	$$a_0 + \frac{a_1}{p_1} + \ldots + \frac{a_n}{p_n} \neq b_0 + \frac{b_1}{p_1} + \ldots + \frac{b_n}{p_n}.$$
\end{proof}

\section{In the Pursuit of Satisfaction}
\label{sec:sat}

The Boolean satisfiability (\SAT) problem for formulae 
was the first problem to be shown \NP-complete by Cook \cite{cook1971complexity}
and Levin \cite{levin1973universal}.
Karp \citep{karp1972reducibility} showed that \SAT is also NP-complete for
formulae in the conjunctive normal form where each clause has at most three literals.
Of course, the same holds for formulae with exactly three literals in each clause.
This is simply because we can introduce a new fresh variable for every missing literal in each clause 
of the given formula without changing its satisfiability.
The set of all formulae with exactly three literals in each clause will denoted by {\em 3-CNF}.
Tovey \citep{tovey1984simplified} showed that \SAT is also NP-complete for 
3-CNF formulae in which each variable occurs at most $4$ times.
We will denote the set of all such formulae by \cnffour.
Schaefer defined in \citep{schaefer1978complexity} the \XSAT problem for 3-CNF formulae 
in which one asks for an truth assignment that makes exactly one literal in each clause true,
and showed it to be \NP-complete.
We define here a new \sameSAT problem for 3-CNF formulae, which 
asks for a valuation that makes exactly the same number of literals
true in every clause (this may be zero, i.e., such a valuation may not make the formula true).
This problem will be a crucial ingredient in the proof of the main result of this paper.

The first step is to show that \XSAT problem is \NP-complete even when restricted to \cnffour formulae.

\begin{theorem}
	\label{thm:xsat}
	The \XSAT problem for \cnffour is \NP-complete.
\end{theorem}
\begin{proof}
	Obviously the problem is in \NP, because we can simply guess a valuation 
	and check how many literals are true in each clause in linear time.
	
	To prove \NP-hardness, we are going to reduce from the \SAT problem for \cnffour, 
	which is NP-complete \citep{tovey1984simplified}.
	Assume we are given a \cnffour formula
	\[\phi = C_1 \wedge C_2 \wedge \ldots \wedge C_m\]
	with $m$ clauses $C_1, \ldots, C_m$ and $n$ propositional variables $v_1, \ldots, v_n$,
	where $C_j = x_j \vee y_j \vee z_j$ for $j = 1,\ldots,m$ and
	each $x_j,y_j,z_j$ is a literal equal to $v_i$ or $\lnot v_i$
	for some $i$.
	We will construct a \cnffour formula $\phi'$ with $3m$ clauses and $n + 4m$ propositional variables 
	such that $\phi$ is satisfiable iff $\phi'$ is an instance of the \XSAT problem.
	This will be based on the construction already given in \cite{schaefer1978complexity}.
	
	The formula $\phi'$ is constructed by replacing each clause in $\phi$ with three new clauses.
	Specifically, the $j$-th clause $C_j = x_j \vee y_j \vee z_j$ is replaced by
	$C'_j := (\lnot x_j \vee a_j \vee b_j) \wedge (b_j \vee y_j \vee c_j) \wedge (c_j \vee d_j \vee \lnot z_j)$
	where $a_j, b_j, c_j, d_j$ are four fresh propositional variables.
	It is quite straightforward to check that only a valuation that makes $C_j$ true 
	can be extended to a valuation that makes exactly one literal true in each of the clauses in $C'_j$.
	Notice that such a constructed $\phi'$ is a \cnffour formula, because this transformation does not 
	increase the number of occurrences of any of the original variables in $\phi$
	and each of the new variables is used at most twice.
	
	Now, if there exists a valuation that makes every clause in $\phi$ true, then as argued above it can be extended to a valuation 
    that makes exactly one literal true in every clause in $\phi'$.
	
	To show the other direction, let $\nu$ be a valuation that makes exactly one literal true in every clause in $\phi'$.
	Consider for every $j = 1,\ldots,m$ the projection of $\nu$ on the set of variables
	occurring in the clause $C_j$. Suppose that such a valuation makes $C_j$ false.
    It follows that it would not be possible to extend this valuation 
    to a valuation that makes exactly one literal true in every clause in $C'_j$. 
    However, we already know that $\nu$ is such a valuation, so this leads to a contradiction.
\end{proof}

Although Theorem \ref{thm:xsat} is of independent interest, all that we need it for is to prove our next theorem. 

\begin{theorem}
    \label{thm:all-same}
	The \sameSAT problem for \cnffour\ \ formulae is \NP-complete.
\end{theorem}
\begin{proof}
	The \sameSAT problem is clearly in \NP, because we can simply guess
	a valuation and check whether it makes exactly the same number of literals
	in every clause true.
	
	To proof \NP-hardness, we reduce from the \XSAT problem for \cnffour formulae (Theorem \ref{thm:xsat}). 
	Let $\phi$ be any \cnffour formula and let us consider a new formula 
	$\phi' = \phi \wedge (x \vee x \vee \lnot x)$,
	where $x$ is a fresh variable that does not occur in $\phi$.
	Notice that $\phi'$ is also a \cnffour formula.
	We claim that $\phi$ is an instance of \XSAT iff $\phi'$ is an instance of \sameSAT.
	
	\smallskip \noindent ($\Rightarrow$) 
	If $\nu$ is a valuation that makes exactly one literal in every clause in $\phi$ true, 
	then extending it by setting $\nu'(x) = \bot$ would
	make exactly one literal in every clause in $\phi'$ true.
	
	\smallskip \noindent ($\Leftarrow$) 
	Let $\nu$ be a valuation that makes the same number of literals in every clause in $\phi'$ true.
	This number cannot possibly be $0$ or $3$, 
	because there is at least one true literal and one false literal in 
	the clause $(x \vee x \vee \lnot x)$.
	
	If $\nu$ makes exactly one literal in every clause in $\phi'$ true, then 
	the same holds for $\phi$.
	
	If $\nu$ makes exactly two literals in every clause true, then 
	consider the valuation $\nu'$ such that $\nu'(y) = \lnot \nu(y)$ for every propositional variable $y$ in $\phi'$.
	Notice that $\nu'$ makes exactly one literal in every clause in $\phi'$ true,
	so the same holds for $\phi$. 	
\end{proof}

\section{Being Rational Makes You Stronger}
\label{sec:strong}

In this section, we build on the results obtained in the previous two section and show
strong NP-hardness of all the decision problems defined in Section \ref{sec:background}.
As a direct consequence, no pseudo-polynomial algorithm can exist for solving any of these
problems unless \Ptime=\NP.
Instead of showing strong NP-hardness for each of these problems separably, 
we will show one ``master'' reduction for the \usubss problem instead.
This reduction will then be reused to show strong NP-hardness of the \partition problem, 
and from these two results the strong NP-hardness of all the other decision problems studied in this paper will follow.

\begin{theorem}
    \label{thm:usbss}
	The \usubss problem with rational weights is strongly NP-complete.
\end{theorem}
\begin{proof}
	For a given instance
	$A = \{w_1, \ldots, w_n\}$ and target weight $W$, 
	we know that the quantities $q_i$, for all $i = 1,\ldots,n$, have to satisfy $q_i \leq W/w_i$.
	This in fact shows that the problem is in \NP, because 
	all the quantities $q_i$ when represented in binary can be written down in polynomial space
    and can be guessed at the beginning. 
	We can then simply verify whether $\sum_{i=1}^n q_i \cdot a_i = W$ holds
	in polynomial time by adding the rational numbers inside this sum 
	one by one (while representing all the numerators and denominators in binary).
	
	To prove strong NP-hardness, we provide a reduction from the \sameSAT problem for \cnffour formulae (which is NP-complete due to Theorem \ref{thm:all-same}).
	Assume we are given a \cnffour formula
	\[\phi = C_1 \wedge C_2 \wedge \ldots \wedge C_m\]
	with $m$ clauses $C_1, \ldots, C_m$ and $n$ propositional variables $x_1, \ldots, x_n$,
	where $C_j = a_j \vee b_j \vee c_j$ for $j = 1,\ldots,m$ and
	each $a_j,b_j,c_j$ is a literal equal to $x_i$ or $\lnot x_i$
	for some $i$.
	For a literal $l$, we write that $l \in C_j$ iff $l$ is equal to $a_j, b_j$ or $c_j$.
	We will now construct a set of items $A$ of size
	polynomial in $n+m$ and a polynomial weight $W$ such that 
    $A$ with the total weight $W$ is a positive instance of \usubss
	iff $\phi$ is satisfiable.
    
	We first need to construct a list of $n+m$ different prime numbers 
	$(p_1,\ldots,p_{n+m})$ that are all larger than $n+5$.
	It suffices to pick $p_i = \pi_{i+n+5}$ for all $i$, 
	because clearly $\pi_j > j$ for all $j$.
	Thanks to Corollary \ref{cor:primes-complexity},
	we can list all these primes numbers in the unary notation
	in time and space polynomial in $n+m$.

	The set $A$ will contain one item per each literal.
	We set the weight of the item corresponding to the literal $x_i$
	to 
	\[1 + \nfrac{1}{p_i} - \nfrac{1}{p_{i \plone{n}}} + \sum_{\{j \mid x_i \in C_j \}} \left(\nfrac{1}{p_{n+j}} - \nfrac{1}{p_{n+j\plone{m}}}\right)\]
	and
	corresponding to the literal $\lnot x_i$ to
	\[1 + \nfrac{1}{p_i} - \nfrac{1}{p_{i\plone{n}}} + \sum_{\{j \mid \lnot x_i \in C_j \}} \left(\nfrac{1}{p_{n+j}} - \nfrac{1}{p_{n+j\plone{m}}}\right).\]

	Notice that each of these weights is $\geq 1 - \frac{5}{p_1} > 0$,
	because each literal occurs at most four times in $\phi$,
	and $p_1 > 5$ is the smallest prime number among $p_i$-s.
	At the same time, all of them are also $\leq 1 + \frac{5}{p_1} < 2$.
	Moreover, they can all be written in unary using polynomial space, 
	because each literal occurs in at most four clauses and 
	so this sum will have at most 11 terms in total.
	We can then combine all these terms into a single rational number.
	Its denominator will be at most equal to $p_{n+m}^{10}$,
	because $p_{n+m}$ is the largest prime number among $p_i$-s and
	the first of these terms is equal to 1.
	The numerator of this rational number has to be smaller than $2p_{n+m}^{10}$,
	because we already showed this number to be $< 2$.
	So both of them will have size $\calO((2n+m)^{10}\log^{10} (2n+m))$
	when written down in unary, because $p_{n+m} = \pi_{2n+m+5} < 2 (2n+m+5)\log (2n+m+5)$
	due to Theorem \ref{thm:prime-dist}.
	Set $A$ has $2n$ such items and so all its elements' weights can be written down in unary
	using $\calO(n(2n+m)^{10}\log^{10} (2n+m))$ space.
	
	Notice that the total weight of $A$ is equal to
	\[2n + \sum_{i=1}^n \left(\nfrac{2}{p_i} - \nfrac{2}{p_{i \plone{n}}}\right) + \sum_{j=1}^m \left(\nfrac{3}{p_{n+j}} - \nfrac{3}{p_{n+j\plone{m}}}\right)\]
	because there are $2n$ literals, each variable corresponds to two literals, and each clause contains exactly three literals.
	As both of the two sums in this expression are telescoping, 
	we get that the total weight is in fact equal to $2n$.
	We claim that the target weight $W = n$ is achievable by picking items from $A$ (each item possibly multiple times)
	iff $\phi$ is a positive instance of \sameSAT.
	
	\smallskip \noindent ($\Rightarrow$) 
	Let $q_i$ and $q'_i$ be the number of times an item corresponding to, respectively, literal $x_i$ and $\lnot x_i$ is chosen
	so that the total weight of all these items is $n$. 
	
	For $i=1,\ldots,n$, we define $t_i := q_i + q_i'$.
	For $j=1,\ldots,m$, we define $t_{n+j}$ to be the number of times an item corresponding to a literal in $C_j$ is chosen.
	For example, if $C_j = x_1 \vee \lnot x_2 \vee x_5$ then $t_{n+j} = q_1 + q'_2 + q_5$.
	Finally, let $T := \sum_{i=1}^n q_i + q'_i$ be the total number of items chosen.
	Notice that $T \leq W / (1 - \frac{5}{p_1}) < W / (1 - \frac{5}{n+5}) = n+5$.
	
	Now the total weight of the selected items can be expressed using $t_i$-s as follows:
	\begin{equation} \tag{$\star$}
	\label{eq:sum}
	\sum_{i=1}^n t_i + \sum_{i=1}^n \frac{t_i - t_{i \minone{n}}}{p_i} + \sum_{j=1}^m \frac{t_{n+j} - t_{n+j \minone{m}}}{p_{n+j}}
	\end{equation}
	Notice that $|t_i - t_{i \minone{n}}| < n+5$ and $p_i > n+5$ for all $i=1,\ldots,n$, and
	$|t_{n+j} - t_{n+j \minone{m}}| < n+5$ and $p_{n+j} > n+5$ for all $j=1,\ldots,m$.
	It now follows from Lemma \ref{lem:uniqueness} that
    (\ref{eq:sum}) can be equal to $n$ if and only if
	$\sum_{i=1}^n t_i = n$, and $t_1=t_2=\ldots=t_n$, and $t_{n+1} = t_{n+2} = \ldots = t_{n+m}$.
	The first two facts imply that for all $i = 1,\ldots,n$, we have $t_i = 1$ and so exactly one item 
	corresponding to either $x_i$ or $\lnot x_i$ is chosen. 
	The last fact states that in each clause exactly the same number of items corresponding to its literals is chosen.
	It is now easy to see that the \sameSAT condition is satisfied by $\phi$ for the valuation $\nu$ such that,
    for all $i \in \{1,\ldots,n\}$, we set
    $\nu(x_i) = \top$ iff $q_i =1$. 

	\smallskip \noindent ($\Leftarrow$) 
	Let $\nu$ be a valuation for which $\phi$ satisfies the \sameSAT condition. 
	We set the quantities $q_i$ and $q'_i$, the number of times an item corresponding to the literal $x_i$ and $\lnot x_i$ is picked, as follows.
	If $\nu(x_i) = \top$ then we set $q_i = 1$ and $q'_i = 0$.
	If $\nu(x_i) = \bot$ then we set $q_i = 0$ and $q'_i = 1$.
	
	Let us define $t_i$-s as before.
	Note that we now have $t_i = 1$ for all $i = 1,\ldots,n$ and 
	$t_{n+1} = t_{n+2} = \ldots = t_{n+m}$, because the \sameSAT condition is satisfied by $\nu$.
	We can now easily see from the expression (\ref{eq:sum}) that the total weight of the just picked items is equal to $n$.
\end{proof}

Although the strong NP-hardness complexity of the \partition problem 
does not follow from the statement of Theorem \ref{thm:usbss}, it follows from its proof as follows.

\begin{theorem}
	The \partition problem with rational weights is strongly NP-complete.
\end{theorem}
\begin{proof}
	Just repeat the proof of Theorem \ref{thm:usbss} without any change. In this case we know {\em a priori} that $q_i \in \{0,1\}$, which does not make any difference to the used reasoning. Notice that the target weight $W$ chosen in the reduction is exactly equal to half of the total weights of all the items in $A$,
    so the \usubss problem instance constructed can also be considered to be a \partition problem instance.
\end{proof}

Now, as the \subss problem is a generalisation of the \partition problem, we instantly get the following result.

\begin{corollary}
	The \subss problem with rational weights is strongly NP-complete.
\end{corollary}

Finally, we observe that the 0-1 \knapsack and \uknapsack problems are generalisations of the \subss and \usubss problems, respectively. 
To see this just restrict the weight and profit of each item to be equal to each other as well as require $V = W$.
Any such instance is a positive instance of 0-1 \knapsack (\uknapsack) if and only if it is a positive instance of \subss (respectively, \usubss).

\begin{corollary}
    The 0-1 \knapsack and \uknapsack problems with rational weights are strongly NP-complete.
\end{corollary}

\section{Approximability}
\label{sec:fptas}

In this section, we briefly discuss the counter-intuitive fact that 
the optimisation version of all the decision problems defined in Section \ref{sec:background} admit a
fully polynomial-time approximation scheme (FPTAS) even though we just showed them to be
strongly NP-complete. First, let us restate a well-known result concerning this.
\begin{corollary}[Corollary 8.6 in \cite{vazirani2013approximation}]
Let $\Pi$ be an NP-hard optimisation problem satisfying the restrictions of Theorem 8.5 in \cite{vazirani2013approximation} (first shown in \cite{garey1978strong}). 
If $\Pi$ is strongly NP-hard, then $\Pi$ does not admit an FPTAS, assuming \Ptime ̸$\neq$ \NP.
\end{corollary}
The crucial assumption made in Theorem 8.5 of \cite{vazirani2013approximation} is that the objective function is integer valued,
which does not hold in our case, so there is no contradiction.

First, let us formally define the optimisation version of some of the decision problems studied.
The optimisation version of the 0-1 \knapsack problem with capacity $W$ 
asks for a subset of items with the maximum possible total profit and whose weight does not exceed $W$.
As for the \subss problem, its optimisation version asks for a subset of items whose total weight is maximal, but $\leq W$.
The optimisation version of the other decision problems from Section \ref{sec:background} can also be defined 
(see, e.g., \cite{books/daglib/0010031}).

Now, let us formally define what we mean by an approximation algorithm for these problems.
We say that an algorithm is a {\em constant factor approximation algorithm} with a {\em relative performance $\rho$}
iff, for any problem instance, $I$, the cost of the solution that it computes, $f(I)$, satisfies:
\begin{itemize}
	\item for a maximisation problem: $(1-\rho)\cdot \OPT(I) \leq f(I) \leq \OPT(I)$
	\item for a minimisation problem: $\OPT(I) \leq f(I) \leq (1+\rho)\OPT(I)$
\end{itemize}
where $\OPT(I)$ is the optimal cost for the problem instance $I$. 
We are particularly interested in polynomial-time approximation algorithms.
A polynomial-time approximation scheme (PTAS) is an algorithm that, for every $\rho > 0$,
runs in polynomial-time and has relative performance $\rho$.
Note that the running time of a PTAS may depend in an arbitrary way on $\rho$.
Therefore, one typically strives to find a fully polynomial-time approximation scheme (FPTAS), which 
is an algorithm that runs in polynomial-time in the size of the input and $1/\rho$.

We will focus here on defining an FPTAS for 0-1 \knapsack problem with rational coefficients.
An FPTAS for the other optimisation problems considered 
in this paper can be defined in essentially the same way
and thus their details are omitted.

\begin{theorem}
	The 0-1 \knapsack problem with rational coefficients admits an FPTAS.
\end{theorem}
\begin{proof}
	We claim that we can simply reuse here any FPTAS for the 0-1 \knapsack problem with integer coefficients.
	Let $I$ be a 0-1 \knapsack instance with rational coefficients.
	We turn $I$ into an instance with integer coefficients only, $I'$,
	by the usual trick of multiplying all the rational coefficients by the least common multiple
	of the denominators of all the rational coefficients in $I$.
	Let us denote this least common multiple by $\alpha$.    
	Assuming that all the coefficients in $I$ are represented in binary, 
	then when multiplying them by $\alpha$ (again in binary representation),
	their size, as argued at the end of Section \ref{sec:background}, will only increase polynomially.
	Therefore, the size of $I'$ is just polynomially larger than $I$. (We should not use the unary notation because then these numbers may grow exponentially.)

	Notice that $\alpha \cdot \OPT(I) = \OPT(I')$.
	In fact, the profit of any subset of items $A$ in $I'$, denoted by $\val'(A)$, is $\alpha$ times bigger
	than the profit of this set of items in $I$, denoted by $\val(A)$.
	Let us now run on $I'$ any FPTAS for 0-1 \knapsack problem with integer coefficients with relative performance $\rho$
	(e.g., \cite{ibarra1975fast}).
	This will return as a solution a subset of items, $B$, such that 
	$\val'(B) \geq (1-\rho) \OPT(I')$.
	This implies that $\val(B) \geq (1-\rho) \OPT(I)$ so the same subset of items $B$ has also
	the same relative performance $\rho$ on the original instance $I$.
\end{proof}

\section{Conclusions}
\label{sec:conclusions}

In this paper we studied how the computational complexity of the \partition, 0-1 \subss, \usubss, 0-1 \knapsack, \uknapsack problems 
changes when items' weights and profits can be any rational numbers.
We showed here, as opposed to the setting where all these values are integers, 
that all these problems are strongly NP-hard, which means that 
there does not exists a pseudo-polynomial algorithm for solving them unless \Ptime=\NP.
Nevertheless, we also showed that all these problem admit an FPTAS
just like in the integer setting.
Finally, we just want to point out that 
if we restrict ourselves to only 
rational weights and profits with a finite representation as decimal numerals,
then these problems are no longer strongly \NP-complete.
This is because we could then simply multiply all the input numbers 
by a sufficiently high power of 10 and get an instance, with
all integer coefficients, whose size is polynomial in the size of the original instance.

\medskip
\noindent {\bf Acknowledgements.\ }
We would like to thank the anonymous reviewers whose comments helped to improve this paper.
This work was partially supported by the EPSRC through grants EP/M027287/1 (Energy Efficient Control) and EP/P020909/1 (Solving Parity Games in Theory and Practice).

\bibliography{biblio}
\bibliographystyle{plain}
\end{document}